\documentclass[11pt]{article}
\usepackage{amsmath,amssymb,amsfonts,amsthm,epsfig}
\usepackage[usenames,dvipsnames]{xcolor}
\usepackage{bm,xspace}
\usepackage{tcolorbox}
\usepackage{cancel}
\usepackage{fullpage}
\usepackage{authorstyle}
\usepackage{framed}
\usepackage{verbatim}
\usepackage{enumitem}
\usepackage{array}
\usepackage{multirow}
\usepackage{afterpage}
\usepackage{mathrsfs}
\usepackage{pifont} 
\usepackage{chngpage}
\usepackage[normalem]{ulem}
\usepackage{boxedminipage}
\usepackage{caption}

\def\colorful{1}

\ifnum\colorful=1

\fi
\ifnum\colorful=0

\fi

\newcommand{\TopDownError}{\textsc{TopDownError}}
\newcommand{\BuildTopDownDT}{\textsc{BuildTopDownDT}}
\newcommand{\bias}{\mathrm{bias}}
\newcommand{\score}{\mathrm{score}} 
\newcommand{\puritygain}{\mathsf{purity\_gain}} 
\renewcommand{\puritygain}{\mathscr{G}\text{-}\mathrm{purity}\text{-}\mathrm{gain}} 
\newcommand{\Tribes}{\textsc{Tribes}}
\newcommand{\Maj}{\textsc{Maj}}
\newcommand{\Rest}{\textsc{Rest}}
\newcommand{\Gimpurity}{\mathscr{G}\text{-}\mathrm{impurity}}

\newcommand{\pparagraph}[1]{\bigskip \noindent {\bf {#1}}}

\makeatletter
\newtheorem*{rep@theorem}{\rep@title}
\newcommand{\newreptheorem}[2]{
\newenvironment{rep#1}[1]{
 \def\rep@title{#2 \ref{##1}}
 \begin{rep@theorem}\itshape}
 {\end{rep@theorem}}}
\makeatother
\theoremstyle{plain}

\newreptheorem{theorem}{Theorem}

\makeatletter
\newtheorem*{rep@claim}{\rep@title}
\newcommand{\newrepclaim}[2]{
\newenvironment{rep#1}[1]{
 \def\rep@title{#2 \ref{##1}}
 \begin{rep@claim}\itshape}
 {\end{rep@claim}}}
\makeatother
\theoremstyle{plain}

\newrepclaim{claim}{Claim}

\begin{document}

\title{Provable guarantees for decision tree induction: \\
the agnostic setting\vspace{10pt}}

\author{Guy Blanc \and \hspace{10pt} Jane Lange \vspace{10pt} \\
\hspace{6pt} {\small {\sl Stanford University}\vspace{10pt}}
\and Li-Yang Tan}

\date{\small{\today}}

\maketitle

\begin{abstract}
We give strengthened provable guarantees on the performance of widely employed and empirically successful {\sl top-down decision tree learning heuristics}.  While prior works have focused on the realizable setting, we consider the more realistic and challenging {\sl agnostic} setting.  We show that for all monotone functions~$f$ and parameters $s\in \N$, these heuristics construct a decision tree  of size $s^{\tilde{O}((\log s)/\eps^2)}$ that achieves error $\le \opt_s + \eps$, where $\opt_s$ denotes the error of the optimal size-$s$ decision tree for $f$.  Previously, such a guarantee was not known to be achievable by any algorithm, even one that is not based on top-down heuristics.  We complement our algorithmic guarantee with a near-matching $s^{\tilde{\Omega}(\log s)}$ lower bound. 
\end{abstract}

\thispagestyle{empty}

\newpage
\setcounter{page}{1}

\section{Introduction}

This paper is motivated by the goal of establishing provable guarantees for the class of popular and empirically successful {\sl top-down decision tree learning heuristics}.  This includes well-known instantiations such as ID3~\cite{Qui86}, its successor C4.5~\cite{Qui93}, and CART~\cite{Bre17}, all widely employed in everyday machine learning applications.  These simple heuristics are also at the heart of more sophisticated decision-tree-based algorithms such as random forests~\cite{Bre01} and XGBoost~\cite{CG16}, which have quickly gained prominence in Kaggle and other data science competitions, and achieve state-of-the-art performance for diverse tasks.  

We will soon formally describe the learning-theoretic framework within which we study these heuristics,  mentioning for now that they build a decision tree $T$ for a binary classifier $f : \R^n \to \zo$ in a {\sl  greedy}, {\sl top-down} fashion: 
\begin{enumerate}
\item Query $\Ind[x_i \ge \theta]$ at the root of $T$, where $x_i$ and $\theta$ are chosen to maximize the {\sl purity gain} 
\[     \mathscr{G}(\E[f]) - \big(\Pr[\bx_i\ge \theta  ]\cdot \mathscr{G}(\E[f_{x_i \ge \theta}])  + \Pr[\bx_i < \theta]\cdot  \mathscr{G}(\E[f_{x_i < \theta }])\big)
\]  
with respect to an {\sl impurity function} $\mathscr{G} : [0,1] \to [0,1]$ and a distribution over labeled examples $(\bx,f(\bx))$.  This carefully chosen function $\mathscr{G}$ encapsulates the {\sl splitting criterion} of the heuristic. 
\item Build the left and right subtrees of $T$ by recursing on $f_{x_i\ge\theta}$ and $f_{x_i<\theta}$ respectively. 
\end{enumerate} 
Different instantiations of this simple approach are distinguished by different impurity functions $\mathscr{G}$, which determine the {\sl order} in which the recursive calls are made.  For example, ID3 and C4.5 uses the binary entropy function $\mathscr{G}(p) = \mathrm{H}(p)$; CART uses the {\sl Gini criterion} $\mathscr{G}(p) = 4p(1-p)$;  Kearns and Mansour proposed and analyzed the function $\mathscr{G}(p) = 2\sqrt{p(1-p)}$~\cite{KM99,DKM96}. 


Even without specifying the impurity function $\mathscr{G}$, it is well known and easy to see that {\sl any} such heuristic can fare poorly even for simple functions~$f$, in the sense of building a decision tree that is much larger than the optimal one.  Consider the most basic setting of binary features, the uniform distribution over inputs, and $f : \{\pm 1\}^n\to \zo$ being the parity of two unknown variables $f(x) = x_j \oplus x_k$ where $j,k\in [n]$.   This function can be computed by a decision tree of size $4$, but since $\E[f_{x_i \ge \theta}] = \E[f_{x_i < \theta}]$ for all $i\in [n]$ and~$\theta$, {\sl any} top-down heuristic---regardless of the impurity function $\mathscr{G}$---may build a complete tree of size $\Omega(2^n)$ before achieving any non-trivial  accuracy.   


\paragraph{Monotonicity and the conjectures of Fiat, Pechyony, and Lee.}  In light of such examples, a question suggests itself: can we identify natural and expressive classes of functions for which strong provable guarantees on the performance of  these top-down heuristics {\sl can} be obtained?  

The first results in this direction were given by Fiat and Pechyony~\cite{FP04,Pec04}, who considered the case of binary features (i.e.~$f : \{\pm 1\}^n \to \zo$) and showed a strong positive result for halfspaces and read-once DNF formulas. For such functions $f$, they showed that these heuristics build a decision tree of {\sl optimal} size that compute~$f$ {\sl exactly}.  Furthermore, and most relevant to our work, they raised the intriguing possibility that {\sl monotonicity} is the key property shared by these functions that enables such a guarantee.\footnote{We consider a function to be monotone if it is either non-decreasing or non-increasing in every coordinate; we do not require the direction to be the same for all coordinates.  (Such functions are sometimes also called ``unate.")}  They conjectured that for all monotone functions~$f$, these heuristics build a decision tree of size ``not far from minimal" that compute $f$ exactly.

Subsequently, Lee~\cite{Lee09} formulated a relaxation of Fiat and Pechyony's conjecture, allowing for an approximate representation rather than an exact representation.  Lee conjectured that for all monotone functions $f : \{\pm 1\}^n \to \zo$ computable by size-$s$ decision trees, these heuristics construct a decision tree of size $\poly(s,1/\eps)$ that achieves error $\eps$ with respect to the uniform distribution over inputs.  (The author further remarked that even a bound that is $\poly(s)$ for constant values of $\eps$ ``would be a huge advance.")

These conjectures of~\cite{FP04} and~\cite{Lee09} are especially appealing because monotonicity, beyond just being a natural assumption that excludes parity and ``parity-like" functions, is an independently important and intensively-studied property in both the theory and practice of machine learning.  Many real-world data sets are naturally monotone in their features. In learning theory, even restricting our attention just to uniform-distribution learning, there is  a large body of work on learning monotone functions~\cite{HM91,KV94,KLV94,Bsh95,BT96,BBL98,Ver98,SM00,Ser04,OS07,Sel08,DLMSWW09,Lee09,JLSW11,OW13,DSFTWW15}.  Partly responsible for this popularity is the fact that monotonicity allows one to sidestep the well-known {\sl statistical query} lower bounds~\cite{BFJKMR94} that hold for many simple concept classes.


\pparagraph{\cite{BLT20}.}  Recent work of the authors  established a weak version of~\cite{Lee09}'s conjecture.  For all monotone functions $f : \{\pm 1\}^n \to \zo$ that are computable by size-$s$ decision trees,~\cite{BLT20} showed that a close variant of these top-down heuristics constructs a decision tree of size $s^{O(\sqrt{\log s}/\eps)}$ that achieves error $\eps$ with respect to the uniform distribution over inputs. 

\cite{BLT20} also showed that the dependence on `$s$' {\sl cannot} be made polynomial, thereby disproving~\cite{Lee09}'s actual conjecture (and~\cite{FP04}'s even stronger conjecture): for all sizes $s \le 2^{\tilde{O}(n^{4/5})}$ and error parameters $\eps \in (0,\frac1{2})$, they exhibited a monotone function $f$ that is computable by a size-$s$ decision tree, and showed that all top-down impurity-based heuristics have to build a decision tree of size $s^{\tilde{\Omega}(\sqrt[4]{\log s})}$ in order to achieve error~$\eps$.

\paragraph{Our contributions.}  We give strengthened provable guarantees on the performance of top-down decision tree learning heuristics.  The three main contributions of our work are: 
\begin{enumerate}
\item We consider the more realistic and challenging {\sl agnostic} setting, where $f$ is an {\sl arbitrary} monotone function.  Prior works focused on the realizable setting, and their results relied on assumptions about the representational complexity of $f$ (i.e.~that $f$ is computable by a small decision tree, or a halfspace, or a read-once DNF formula, etc.).

\item We establish provable guarantees that apply to all top-down heuristics, including ID3, C4.5, and CART.~\cite{BLT20}'s analysis, on the other hand, dealt with a specific {\sl variant} of these heuristics, one whose splitting criterion does not correspond to any impurity function $\mathscr{G}$.   (As a secondary contribution, we further show that~\cite{BLT20}'s guarantees in the realizable setting also hold for all top-down heuristics.) 
\item Our analysis extends to classifiers for {\sl real-valued} features (i.e.~$f : \R^n \to \zo$) and arbitrary product distributions over inputs, whereas prior works dealt with classifiers for binary features (i.e~$f:\{\pm 1\}^n \to \zo$) and mostly focused on the uniform distribution over inputs.  Trees for real-valued features branch on queries of the form $\Ind[x_i \ge \theta]$ for some $\theta \in \R$, whereas trees for binary features branch on queries of the form $\Ind[x_i = 1]$. 
\end{enumerate}  
\newpage 
Our main result is as follows: 

\begin{theorem}[Our main result; informal] 
\label{thm:main-informal} 
Let $f : \R^n \to \zo$ be a monotone function and $\mathcal{D}$ be a product distribution over $\R^n$.  For $s \in \N$, let $\opt_s$ denote the error of the best size-$s$ balanced decision tree for $f$ with respect to $\mathcal{D}$.\footnote{A balanced decision tree of size $s$ is one that has depth $O(\log s)$. This technical assumption is not necessary in the case of binary features and the uniform distribution over inputs (see \Cref{thm:main}); it is only necessary for the general setting of real-valued features and arbitrary product distributions.} Let $\mathscr{G}$ be any impurity function.  For $t =  s^{\tilde{\Theta}((\log s)/\eps^2)}$, the size-$t$ decision tree for $f$ constructed by the top-down heuristic with $\mathscr{G}$ as its splitting criterion achieves classification error $\le \opt_s + \eps$ with respect to $\mathcal{D}$. 
\end{theorem} 

Previously, such a guarantee was not known to be achievable by any algorithm, even one that is not based on top-down heuristics (i.e.~\Cref{thm:main-informal} represents the first algorithm for properly learning decision trees in an agnostic setting).

We complement~\Cref{thm:main-informal} with a near-matching lower bound.  We show that for all $s \le 2^{\tilde{O}(\sqrt{n})}$, there is a monotone function $f$ such that $\opt_s \le 0.01$, and yet any top-down heuristic has to grow a tree of size $s^{\tilde{\Omega}(\log s)}$ to even achieve error $\le 0.49$ with respect to $f$.  Taken together with~\cite{BLT20}'s $s^{O(\sqrt{\log s})}$ {\sl upper} bound in the realizable setting, this exhibits a separation between the realizable and agnostic settings.

\subsection{Formal statements of our results} 
\label{sec:formal-statements} 

We define a {\sl partial tree} to be a decision tree with unlabeled leaves, and write $T^{\circ}$ to denote such trees.  We refer to any decision tree $T$ obtained from $T^{\circ}$ by a labeling of its leaves as a {\sl completion} of $T^{\circ}$.  Given a tree a partial tree $T^\circ$ and a function $f$, there is a canonical completion of $T^\circ$ that minimizes the approximation error with respect to $f$: 

\begin{definition}[$f$-completion of a partial tree] 
Let $T^\circ$ be a partial tree and $f : \R^n \to \zo$.  Consider the following completion of $T^\circ$:  for every leaf $\ell$ in $T^\circ$, label it $\mathrm{round}(\E[f_\ell])$, where $f_\ell$ denotes the restriction of $f$ by the path leading to $\ell$ and $\mathrm{round}(p) = 1$ if $p\ge \frac1{2}$ and $0$ otherwise.  This completion minimizes the approximation error $\Pr[T(\bx) \ne f(\bx)]$; we refer to it as the \emph{$f$-completion of $T^\circ$} and denote it as $T^\circ_f$. 
\end{definition}

\begin{definition}[Impurity functions and strong concavity]
\label{def:impurity}
An {\sl impurity function} $\mathscr{G} : [0,1] \to [0,1]$ is a concave function that is symmetric around $\frac1{2}$, and satisfies $\mathscr{G}(0) = \mathscr{G}(1) = 0$ and $\mathscr{G}(\frac1{2}) = 1$.  We say that $\mathscr{G}$ is {\sl $\kappa$-strongly concave} if for all $a,b \in [0,1]$, 
\[ \frac{\mathscr{G}(a)+\mathscr{G}(b)}{2}  \le \mathscr{G}\left(\frac{a+b}{2}\right) - \frac{\kappa}{2}\cdot (b-a)^2. \] 
\end{definition} 

\begin{remark}[$\kappa$ values of common impurity functions]
ID3 and C4.5 use binary entropy as their impurity function, which is $\kappa$-strongly concave for $\kappa = 1/\ln(2)$, or $\approx 1.4$. Gini impurity, which is used by CART, is strongly concave for $\kappa = 2$, and \cite{KM99}'s impurity function is strongly concave for $\kappa = 1$.
\end{remark}

We now formally define the top-down heuristics that we study.  
\begin{definition}[$\mathscr{G}$-impurity of a partial tree]
Let $\mathscr{G} : [0,1] \to [0,1]$ be an impurity function and $\mathcal{D}$ be a distribution over $\R^n$.  For $f : \R^n \to \zo$ and a partial tree $T^\circ$, the {\sl $\mathscr{G}$-impurity of $T^\circ$ with respect to $f$} is defined to be 
\[     \Gimpurity_{f,\mathcal{D}}(T^\circ) \coloneqq \sum_{\text{leaves $\ell \in T^\circ$}} \Prx_{\bx \sim \mathcal{D}}[\,\text{$\bx$ reaches $\ell$}\,] \cdot \mathscr{G}(\E[f_\ell]).
\]
\end{definition} 

 If $T^\circ$ is a partial tree and $\ell$ is a leaf of $T^\circ$, we write $T^\circ_{\ell,\Ind[x_i\ge \theta]}$ to denote the extension of $T^\circ$ obtained by splitting $\ell$ with a query to $\Ind[x_i\ge \theta]$.  The following algorithm captures the top-down decision tree learning heuristics that we study in this work:


\begin{figure}[H]
  \captionsetup{width=.9\linewidth}
\begin{tcolorbox}[colback = white,arc=1mm, boxrule=0.25mm]
\vspace{3pt} 

$\textsc{BuildTopDownDT}_{\mathscr{G},\mathcal{D}}$($f,t$):  \vspace{6pt} 

\ \ Initialize $T^\circ$ to be the empty tree. \vspace{4pt} 

\ \ while ($\size(T^\circ) < t$) \{  \\

\vspace{-15pt}
\begin{itemize}
    \item[] Grow $T^\circ$ by splitting leaf $\ell$ with a query to $\Ind[x_i \ge \theta]$, where $\ell$ and $\Ind[x_i\ge \theta]$ maximize:
    \[
        \Gimpurity_{f,\mathcal{D}}(T^\circ) 
    - \Gimpurity_{f,\mathcal{D}}(T^\circ_{\ell,\Ind[x_i\ge \theta]}),
    \]
    \item[] the {\sl purity gain} with respect to $\mathscr{G}$ and $\mathcal{D}$.  \\
\end{itemize}
\vspace{-18pt}

\ \ \}  \vspace{4pt} 

\ \ Output the $f$-completion of $T^\circ$.

\vspace{3pt}
\end{tcolorbox}
\caption{Top-down heuristic for building a size-$t$ decision tree approximation of a function~$f : \R^n \to \zo$, using the impurity function $\mathscr{G}$ as its splitting criterion.}
\label{fig:TopDown}
\end{figure}


We write $\TopDownError_{\mathscr{G},\mathcal{D}}(f,t)$ to denote the error  $\Pr[T(\bx)\ne f(\bx)]$, where $T$ is the size-$t$ tree constructed by $\BuildTopDownDT_{\mathscr{G},\mathcal{D}}(f,t)$.

\begin{definition}[$\opt_s$] 
For a function $f : \R^n \to \zo$, a distribution $\mathcal{D}$ over $\R^n$, and an integer $s \in \N$, we write $\opt_{f,\mathcal{D},s} \in [0,\frac1{2}]$ to denote the error of the best size-$s$ decision tree for $f$: 
\begin{equation*} \opt_{f,\mathcal{D},s} \coloneqq \min\Big\{ \Prx_{\bx\sim\mathcal{D}}[T(\bx) \ne f(\bx)] \colon \text{$T$ is a size-$s$ decision tree}\,\Big\}. 
\end{equation*}
When $f$ and $\mathcal{D}$ are clear from context, we simply write $\opt_s$. 
\end{definition} 

\pparagraph{Our main algorithmic guarantee.} We first state and prove our results in the setting of binary features and the uniform distribution over inputs.  As alluded to in the introduction, in~\Cref{section:real trees} we will show how our results in this specialized setting extend to the more general setting of real-valued features and arbitrary product distributions over inputs.

\begin{theorem}[Our main algorithmic guarantee for binary features and the uniform distribution]
\label{thm:main}  Let $f : \{\pm 1\}^n \to \{ 0,1\}$ be a monotone function, $\mathcal{U}$ be the uniform distribution over $\{ \pm 1\}^n$, and $\mathscr{G} : [0,1]\to [0,1]$ be an $\kappa$-strongly concave impurity function.  For all $s \in \N$ and $\eps \in (0,\frac1{2})$, 
\[ \TopDownError_{\mathscr{G},\mathcal{U}}(f,s^{\Theta(\log s)/\kappa  \eps^2}) \le \opt_s + \eps. \] 
\end{theorem}


In words,~\Cref{thm:main} says that for all $s \in \N$, the decision tree of size $s^{\Theta(\log s)}$ constructed by $\BuildTopDownDT$ achieves error that nearly matches that of the best size-$s$ decision tree for~$f$.


\pparagraph{A near-matching lower bound.}  We contrast~\Cref{thm:main} with~\cite{BLT20}'s result in the {\sl realizable} setting: if $f$ is a monotone function that is computable by a size-$s$ decision tree (i.e.~$\opt_s = 0$), then a variant of the top-down heuristics grows a tree of size $s^{\Theta(\sqrt{\log s})}$ that achieves error $\eps$.\footnote{The heuristic that~\cite{BLT20} analyzes does not correspond to any impurity function $\mathscr{G}$. As mentioned in the introduction, in~\Cref{section: realizable setting} we show that~\cite{BLT20}'s guarantees on their variant of the top-down heuristics in fact hold for all top-down heuristics.}  With this in mind, it is natural to wonder if the parameters of~\Cref{thm:main} can be improved to $\TopDownError_{\mathscr{G},\mathcal{U}}(f, s^{\Theta_{\kappa,\eps}(\sqrt{\log s})}) \le \opt_s + \eps.$
We complement~\Cref{thm:main} with a lower bound that rules out such an improvement. We show that the dependence on `$s$' in~\Cref{thm:main} is in fact near-optimal:   

\begin{theorem}[Our main lower bound] 
\label{thm:lower-bound}
For all $s \le 2^{\tilde{O}(\sqrt{n})}$, there is a monotone function $f : \{ \pm 1\}^n \to \zo$ such that $\opt_s \le 0.01$ with respect to the uniform distribution over inputs, and $\TopDownError_{\mathscr{G},\mathcal{U}}(f,s^{\tilde{\Omega}(\log s)}) \ge 0.49$ for any impurity function~$\mathscr{G}$.  
\end{theorem}

Taken together with~\cite{BLT20}'s $s^{O(\sqrt{\log s})}$ upper bound in the realizable setting,~\Cref{thm:lower-bound} exhibits a separation between the realizable and agnostic settings.

\subsection{Related work} 

Kearns and Mansour~\cite{Kea96,KM99} were the first to study top-down decision tree learning heuristics using the framework of learning theory.  They showed that these heuristics can be viewed as {\sl boosting algorithms}, where one views the functions queried at the internal nodes of the tree  (single variables in our case) as weak hypotheses.    As is standard in results on boosting, their results are {\sl conditional} in nature: they assume the existence of weak hypotheses for filtered-and-rebalanced versions of the original distribution (what they call ``The Weak Hypothesis Assmption"), and they show how these top-down heuristics build a decision tree that combines these weak hypotheses into a strong one.  Dietterich, Kearns, and Mansour~\cite{DKM96} gave an experimental comparison of the impurity functions used by ID3, C4.5, and CART, along with a new impurity function that~\cite{KM99} had proposed. 

Fiat and Pechyony~\cite{FP04,Pec04} studied functions~$f$ that are computable by linear threshold functions or read-once DNF formulas, and showed that these heuristics build a decision tree of optimal size that compute $f$ exactly.  Recent work of Brutzkus, Daniely, and Malach~\cite{BDM19} studies functions $f$ that are conjunctions and read-once DNF formulas, and gives theoretical and empirical evidence showing that for such functions, the variant of ID3 proposed by~\cite{KM99}, when run for $t$ iterations, grows a tree that achieves accuracy that matches or nearly matches that of the best size-$t$ tree for $f$.  The same authors~\cite{BDM19-smooth} also show that ID3 learns $(\log n)$-juntas in the setting of smoothed analysis. 

\pparagraph{Learning algorithms not based on top-down heuristics: improper algorithms for learning decision trees.} O'Donnell and Servedio~\cite{OS07} gave a $\poly(n,s^{1/\eps^2})$-time uniform-distribution algorithm for learning monotone functions computable by size-$s$ decision trees.  This remains the fastest algorithm for the realizable setting.~\cite{OS07}'s algorithm is not based on the top-down heuristics that are the focus of our work (and the others discussed above); indeed, their algorithm does not output a decision tree as its hypothesis (i.e.~it is not a proper learning algorithm).  For the agnostic setting, the results of Kalai, Klivans, Mansour, and Servedio~\cite{KKMS08} can be used to give a uniform-distribution algorithm that runs in time $n^{O(\log(s/\eps))}$ and outputs a hypothesis that achieves error $\opt_s + \eps$.  Compared to~\Cref{thm:main}, this algorithm does not require $f$ to be monotone, but like~\cite{OS07}'s algorithm it is also improper.  The work of~\cite{GKK08} gives a uniform-distribution algorithm that runs in $\poly(n,s,1/\eps)$ time; however, their algorithm requires the use of {\sl membership queries}, and is also improper.  Furthermore, all the results discussed in this paragraph only hold in the setting of binary features and with respect to the uniform distribution over inputs.

\subsection{Preliminaries} 

We use {\bf boldface} (e.g.~$\bx \sim \R^n$) to denote random variables.  Given two functions $f,g : \R^n \to \zo$, we write $\dist(f,g) \coloneqq \Pr[f(\bx)\ne g(\bx)]$ to denote the distance between $f$ and $g$.  We write $\bias(f) \coloneqq \min\{ \Pr[f(\bx)=0],\Pr[f(\bx)=1]\}$ to denote the distance of $f$ to the closest constant function; equivalently, $\bias(f) = \Pr[f(\bx) \ne \mathrm{round}(\E[f])]$, where $\mathrm{round}(p) = 1$ if $p \ge \frac1{2}$ and $0$ otherwise. If $\ell$ is a leaf in a decision tree, we write $|\ell|$ to denote the depth of $\ell$ within the tree.

\begin{definition}[Monotone functions]
\label{def:monotone} 
We say that a function $f : \R^n \to \zo$ is {\sl monotone} if for all coordinates $i\in [n]$, either 
\begin{itemize}[leftmargin=0.5cm]
\item[$\circ$] $f$ is non-decreasing in the $i$-th direction: $f(x) \le f(y)$ for all $x,y \in \R^n$ such that $x_i \le y_i$, or 
\item[$\circ$] $f$ is non-increasing in the $i$-th direction: $f(x) \ge f(y)$ for all $x,y \in \R^n$ such that $x_i \le y_i$. 
\end{itemize} 
\end{definition}

\section{Binary features and the uniform distribution: \Cref{thm:main}}
\label{sec:main} 

Recall that~\Cref{thm:main} is concerned with the special case of binary features and the uniform distribution over inputs.  Therefore the decision trees that we reason about in the proof of~\Cref{thm:main} split on queries of the form $\Ind[x_i = 1]$, rather than queries of the form $\Ind[x_i \ge \theta]$ as in the general setting of real-valued features.  Also, all probabilities and expectations in this proof are with respect to the uniform distribution over inputs. 

\begin{definition}[Influence] 
Given a function $f : \{ \pm 1 \}^n \to \zo$, the {\sl influence of coordinate $i\in [n]$ on $f$} is defined to be 
\[ \Inf_i(f) \coloneqq \Pr[f(\bx)\ne f(\bx^{\oplus i})], \]
where  $\bx^{\oplus i}$ denotes $\bx$ with its $i$-coordinate flipped, and $\bx \sim \{ \pm 1\}^n$ is uniform random.  The {\sl total influence} of $f$ is defined to be $\Inf(f) \coloneqq \sum_{i=1}^n \Inf_i(f)$. 
\end{definition} 

 A key technical ingredient in our proof will be an inequality of Jain and Zhang~\cite{JZ11}, a robust version of the powerful O'Donnell--Saks--Schramm--Servedio inequality from the Fourier analysis of boolean functions~\cite{OSSS05}.  The following is a special case of Corollary 1.4 of~\cite{JZ11}: 

\begin{theorem}[\cite{JZ11}'s robust OSSS inequality]
\label{thm:JZ}
Let $f : \{\pm 1\}^n \to \{ 0,1\}$ be any function, and $g : \{\pm 1\}^n \to \{0,1\}$ be a size-$s$ decision tree.  Then 
\[ \max_{i\in [n]}\{\Inf_i(f) \} \ge  \frac{\bias(f)-\dist(f,g)}{\log s}.\]
\end{theorem}

\begin{remark}[Context for~\Cref{thm:JZ}] 
\label{rem:JZ} 
The original OSSS inequality essentially corresponds to the special case of~\Cref{thm:JZ} where $f \equiv g$: if $f$ is a size-$s$ decision tree, then 
\[ \max_{i\in [n]} \{ \Inf_i(f) \} \ge \frac{\Var(f)}{\log s}. \] 
The OSSS inequality can be viewed as a variant of the famed Kahn--Kalai--Linial inequality~\cite{KKL88}, one that takes into account the computational complexity of $f$.   In~\cite{OSSS05} the authors also gave a robust version of their inequality that is qualitative similar to~\Cref{thm:JZ} (see the discussion following Theorem 3.2 of~\cite{OSSS05}): under the assumptions of~\Cref{thm:JZ}, 
\[ \max_{i\in [n]}\{\Inf_i(f) \} \ge  \frac{\Var(f)-2\cdot \dist(f,g)}{\log s}.\]
  This inequality can be used in place of~\Cref{thm:JZ} to prove a statement that is qualitatively similar to~\Cref{thm:main}, but with a weaker error bound of $O(\opt_s) + \eps$ rather than the $\opt_s + \eps$ of~\Cref{thm:main}. 
\end{remark}




\begin{fact}[Influence $\equiv$ correlation for monotone functions]
\label{prop:influence-correlation-monotone}
Let $f : \{\pm 1\}^n \to \{0,1\}$ be a monotone function.  Then $\Inf_i(f) = 2 \cdot |\E[f(\bx)\bx_i]|$ for all $i\in [n]$.  
\end{fact} 


\begin{proposition}[Influence and purity gain]
    \label{prop: score reduction}
    For all $\kappa$-strongly concave impurity heuristics $\mathscr{G}$, all monotone functions $f:\{\pm 1\}^n \rightarrow \zo$, and all coordinates $i \in [n]$,
    \begin{align*}
        \Ex_{\bb\sim \{\pm 1\}}[\mathscr{G}(\E[f_{x_i=\bb}])] \leq
        \mathscr{G}(\E[f])  - \frac{\kappa}{32} \cdot \Inf_i(f)^2.
    \end{align*}
\end{proposition}
\begin{proof}
    Note that $\E[f(\bx)\bx_i] = \lfrac{1}{2}(\E[f_{x_i=1}] - \E[f_{x_i={-1}}])$
    and that $\E[f] = \lfrac{1}{2}(\E[f_{x_i=1}] + \E[f_{x_i={-1}}]).$
    The desired result therefore holds as a direct consequence of the $\kappa$-strong concavity of $\mathscr{G}$ and \Cref{prop:influence-correlation-monotone}.
\end{proof}




\subsection{Proof of~\Cref{thm:main}} 

Let $f : \{ \pm 1\}^n \to \zo$ be a monotone function and $g$ be a size-$s$ decision tree for which $\dist(f,g) = \opt_s$. Fix some $\kappa$-strongly concave impurity heuristic $\mathscr{G}$. Given a partial tree $T^\circ$ (which we should think of as the approximator for $f$ that is being built by $\BuildTopDownDT_{\mathscr{G},\mathcal{U}}$), we define the potential function:  
\[\Gimpurity_{f}(T^\circ) \coloneqq \sum_{\text{leaves $\ell \in T^\circ$}} 2^{-|\ell|} \cdot \mathscr{G}(\E[f_\ell]).  \] 
The next lemma records a few useful properties of this potential function $\Gimpurity_{f}$. (This lemma can be viewed as our analogue of~\cite{BLT20}'s Lemma 5.1.~\cite{BLT20} worked with a potential function that is a variant of ours which has $\Inf(f_\ell)$ in place of $\mathscr{G}(\E[f_\ell])$.)

\begin{definition}[Purity gain]
    Let $f: \{\pm 1\}^n \rightarrow \zo$ be a function and $\mathscr{G}$ be an impurity heuristic.  For a partial tree $T^\circ$ and a leaf $\ell$ of $T^\circ$, we write $T^{\circ}_{\ell,x_i}$ to denote the extension of $T^\circ$ obtained by splitting $\ell$ with a query to $x_i$.  The {\sl purity gain associated with splitting $\ell$ with a query to $x_i$} is defined as
\[          \puritygain_{f}(T^\circ,\ell, x_i) \coloneqq \Gimpurity_{f}(T^\circ) - \Gimpurity_{f}(T^\circ_{\ell, x_i}).
\] 
\end{definition}

It is easy to verify that given a leaf $\ell$, the variable associated with the largest purity gain is exactly one that is most correlated with $f_\ell$:

\begin{proposition}[Proposition 7.7 from \cite{BLT20}]
    For any leaf $\ell$ of $T^\circ$, let $x_{i^\star}$ be the variable that maximizes $\puritygain_\mathscr{G}(\ell, x_i)$ among all $i \in [n]$. Then, 
    \begin{align*}
        \E[f_\ell(\bx)\bx_{i^\star}] \geq \E[f_\ell(\bx)\bx_j] \quad \text{for all $j\in [n]$}.
    \end{align*}
\end{proposition}

\begin{lemma}[Useful properties of $\Gimpurity_{f}$]  \label{lem:useful-properties} \ 
\begin{enumerate} 
\item $\Gimpurity_{f}(\text{\sl empty tree}) = \mathscr{G}(\E[f]) \le 1$. 
\item $\dist(f,T^\circ_f) \le \Gimpurity_{f}(T^\circ)$.
\item For any leaf $\ell$ of $T^\circ$ and variable $i \in [n]$, 
\[ \puritygain_{f}(T^\circ,\ell, x_i) \geq 2^{-|\ell|} \cdot \frac{\kappa}{32} \cdot \Inf_i(f_\ell)^2. \] 
\end{enumerate} 
\end{lemma} 

\begin{proof} 
The first claim follows from the definition of $\Gimpurity_{f}$.  For the second claim, we have that 
\begin{align*} 
\dist(f,T^\circ_f) &= \sum_{\text{leaves $\ell \in T^\circ$}} \Pr[\,\text{$\bx$ reaches $\ell$}\,] \cdot \bias(f_\ell)  \\
&= \sum_{\text{leaves $\ell \in T^\circ$}} 2^{-|\ell|} \cdot \bias(f_\ell) \\
&\le \sum_{\text{leaves $\ell \in T^\circ$}} 2^{-|\ell|} \cdot \mathscr{G}(\E[f_\ell]) \tag*{(\Cref{def:impurity})} \\
&=\Gimpurity_{f}(T^\circ). 
\end{align*} 
As for the third claim, we have that 
\begin{align*} 
\puritygain_f(T^\circ,\ell, x_i) &= 2^{-|\ell|} \left(\mathscr{G}(\E[f_\ell]) - \Ex_{\bb\sim \{\pm 1\}}[\mathscr{G}(\E[(f_\ell)_{x_i = \bb}]) \right) \\
&\geq 2^{-|\ell|} \cdot \frac{\kappa}{32} \cdot \Inf_i(f_\ell)^2, 
\end{align*} 
where the final inequality is by~\Cref{prop: score reduction}.
\end{proof}

The following simple fact states that the error of the $f$-completion of a partial tree $T^\circ$ cannot  increase with further splits: 

\begin{fact}[Splits cannot increase error] 
\label{prop:error-goes-down} 
Let $T^\circ$ be a partial tree, and $\wt{T}^\circ$ be the partial tree that results from splitting a leaf of $T^\circ$.  Then $\dist(f,\wt{T}^\circ_f) \le \dist(f,T^{\circ}_f)$. 
\end{fact} 


We are now ready to prove~\Cref{thm:main}.  By~\Cref{prop:error-goes-down}, it suffices to show that error $\opt_s + \eps$ is achieved after at most $s^{O(\log s)/\eps^2}$ iterations/splits.  We do so by lower bounding the score of the leaf that is split by $\BuildTopDownDT$ in each iteration {\sl before} error $\opt_s + \eps$ is achieved.  Let $T^\circ$ be the size-$(j+1)$ partial tree that is built by $\BuildTopDownDT$ after $j$ iterations.  Suppose $\dist(f,T^{\circ}_f) > \opt_s + \eps$, and let $\ell^\star$, $x_i^\star$ be the leaf and variable of $T^\circ$ that is split in the $(j+1)$-st iteration.  We claim that  
\begin{equation}\label{eq:score-lb} \puritygain_f(T^\circ,\ell^\star,x_{i^\star}) >  \frac{\kappa \cdot \eps^2}{32 \cdot (j+1)(\log s)^2}.
\end{equation} 
Writing $x_{i(\ell)}$ to denote variable with the highest correlation with $f_\ell$, we have that 
\begin{align*} 
&\sum_{\text{leaves $\ell \in T^\circ$}} \puritygain_f(T^\circ,\ell, x_{i(\ell)}) \\
&\ge \sum_{\text{leaves $\ell \in T^\circ$}} 2^{-|\ell|} \cdot \frac{\kappa}{32} \cdot \Inf_{i(\ell)}(f_\ell)^2 \tag*{(\Cref{lem:useful-properties})} \\
&\ge \sum_{\text{leaves $\ell \in T^\circ$}} 2^{-|\ell|} \cdot \frac{\kappa}{32} \cdot \left(\frac{\bias(f_\ell) - \dist(f_\ell,g_\ell)}{\mathrm{size}(g_\ell)}\right)^2 \tag*{(\Cref{thm:JZ})} \\
&\ge \frac{\kappa}{32} \cdot \sum_{\text{leaves $\ell \in T^\circ$}} 2^{-|\ell|} \cdot \left(\frac{\bias(f_\ell) - \dist(f_\ell,g_\ell)}{\log s}\right)^2 \tag*{($\size(g_\ell)\le \size(g) =s$)} \\
&\ge \frac{\kappa}{32}\cdot \left[\,\sum_{\text{leaves $\ell \in T^\circ$}} 2^{-|\ell|} \cdot \left(\frac{\bias(f_\ell) - \dist(f_\ell,g_\ell)}{\log s}\right)\right]^2 \tag*{(Jensen's inequality)} \\
&= \frac{\kappa}{32} \cdot \frac1{(\log s)^2} \Bigg(\,\sum_{\text{leaves $\ell \in T^\circ$}} 2^{-|\ell|}\cdot \bias(f_\ell) \\&\hspace{30mm}- \sum_{\text{leaves $\ell \in T^\circ$}} 2^{-|\ell|}\cdot \dist(f_\ell,g_\ell)\Bigg)^2 \\
&= \frac{\kappa}{32}\cdot \frac1{(\log s)^2} \cdot \left(\dist(f,T^\circ_f) - \dist(f,g) \right)^2 \\
&> \frac{\kappa}{32}\cdot \frac1{(\log s)^2} \cdot \left((\opt_s + \eps) - \opt_s\right)^2 \ = \ \frac{\kappa}{32} \cdot \left(\frac{\eps}{\log s}\right)^2.  
\end{align*} 
It follows that there must be at least one leaf and variable with purity gain greater than $\kappa \eps^2/32(j+1)(\log s)^2$.  Since $\BuildTopDownDT$ splits the leaf and variable with the largest purity gain, this establishes~\Cref{eq:score-lb}.  

Writing $\wt{T}^\circ$ to denote the partial tree that is obtained after $\BuildTopDownDT$ makes the single split with largest purity gain,  we have that
\begin{align*}
 \Gimpurity_{f}(\wt{T}^\circ) &= \Gimpurity_{f}(T^\circ) - \puritygain_f(T^\circ,\ell^\star,x_{i^\star}) \\
 &\le \Gimpurity_{f}(T^\circ) - \frac{\kappa \cdot \eps^2}{32 \cdot(j+1)(\log s)^2}. 
 \end{align*} 
Combining this with the first and second claims of~\Cref{lem:useful-properties}, we have the following: the value of the potential function starts off at at most $1$ with $T^\circ$ being the empty tree, decreases by at least $\kappa \cdot \eps^2/32j(\log s)^2$ with the $j$-th split, and error $\opt_s + \eps$ is achieved once this value drops below $\opt_s + \eps$.  Therefore, we can bound the number of splits necessary to ensure error $\opt_s + \eps$ by the smallest $t$ that satisfies: 
\[ \sum_{j=1}^t \frac{\kappa\eps^2}{32j(\log s)^2} \ge 1-(\opt_s + \eps). \] 
Since 
\[ \sum_{j=1}^t \frac{\kappa\eps^2}{32j(\log s)^2} \ge \frac{\kappa\eps^2 \log t}{32(\log s)^2}, \] 
we conclude that $t \le s^{O(\log s)/\kappa\eps^2}$ suffices.  This completes the proof of~\Cref{thm:main}. 

\section{Real-valued features and product distributions:~\Cref{thm:main-informal}}
\label{section:real trees}
We will prove an extension of \Cref{thm:main} that applies to functions of real-valued features and arbitrary product distributions over inputs.

In order to state our result, we need to slightly restrict our definition of $\opt$ so it only considers ``balanced" trees.

\begin{definition}[Balanced tree]
    A decision tree $T$ of depth $d$ and size $s$ is {\sl balanced} if $ d = O(\log s)$.
\end{definition}

\begin{definition}[$\balancedopt_s$] 
For a function $f : \mathbb{R}^n \to \{0,1\}$, product distribution $\mathcal{D}$ over $\mathbb{R}^n$, and an integer $s \in \N$, we write $\balancedopt_{f,\mathcal{D},s} \in [0,\frac1{2}]$ to denote the error of the best balanced size-$s$ decision tree for $f$: 
\[     \balancedopt_{f,\mathcal{D},s}\coloneqq \min\big\{ \Pr_{\bx \sim \mathcal{D}}[\,T(\bx) \ne f(\bx)\,] \colon \text{$T$ is a balanced size-$s$ decision tree}\,\big\},
\]  When $f$ and $\mathcal{D}$ are clear from context, we simply write $\balancedopt_s$. 
\end{definition} 

We can now formally state our main theorem.
\begin{reptheorem}{thm:main-informal}
    Let $f:\mathbb{R}^n \rightarrow \{0,1\}$ be a monotone function and $\mathcal{D}$ be a product distribution over~$\mathbb{R}^n$. For any $\kappa$-strongly concave impurity heuristic $\mathscr{G}$ and $s \in \mathbb{N}$, 
\[          \TopDownError_{\mathscr{G},\mathcal{D}}(f, s^{\tilde{O}((\log s) / \kappa \varepsilon^2)}) \leq \balancedopt_s + \varepsilon.
    \]
\end{reptheorem}
The proof of \Cref{thm:main-informal} will follow the same overall structure as our proof of \Cref{thm:main}. One key new ingredient is a generalization of \Cref{thm:JZ} to real-valued features; this extension could be of independent interest: 

\begin{theorem}[Extension of \cite{JZ11} to trees for real-valued features]
\label{thm:JZ real input}
Let $f : \mathbb{R}^n \to \zo$ be a monotone function, $\mathcal{D}$ be an arbitrary product distribution over $\mathbb{R}^n$, and $T$ be a size-$s$ balanced decision tree. Then, there exist $i^\star \in [n]$ and $\theta^\star \in \mathbb{R}$ for which $\Pr_{\bx \sim D}[\bx_{i^\star} \geq \theta^\star] = \lfrac1{2}$ and
\begin{align*}
    \Ex_{\bx \sim \mathcal{D}}\big[f(\bx)\cdot \Ind[\bx_{i^\star} \geq \theta^\star]\big] \geq \Omega\left(\frac{\varepsilon}{\log(s) \log\log(s/\eps)}\right)
\end{align*}
where $\varepsilon \coloneqq \bias(f) - \dist(f,T)$, and where $\bias(f)$ and $\dist(f,T)$ are also measured with respect to~$\mathcal{D}$. 
\end{theorem}

To prove \Cref{thm:main-informal}, we apply \Cref{thm:JZ real input} in the same way \Cref{thm:JZ} is used to prove \Cref{thm:main}. 

\subsection{Proof of \Cref{thm:main-informal}}

We first prove~\Cref{thm:JZ real input}:
\begin{proof}[Proof of~\Cref{thm:JZ real input}]
First, without loss of generality, we can assume that $\mathcal{D}$ is the uniform distribution on $[0,1]^n$. Otherwise, we can transform each variable by its CDF. This also means that we will set $\theta^\star = \frac1{2}$.

We will use  \Cref{thm:JZ} as a black box in our proof of \Cref{thm:JZ real input}. In order to do this, we will relate $T$, a decision tree with real-valued inputs, to a specially constructed decision tree with boolean-valued inputs. Each input of $T$ will be encoded into a boolean vector of width $w = O(\log(s/\varepsilon))$ using the encoder $E:[0,1] \rightarrow \{\pm 1\}^w$ defined as follows.
\begin{align*}
    E(x) = \textsc{Binary}(\lfloor x \cdot 2^w \rfloor),
\end{align*}
where $\textsc{Binary}$ is the function that encodes an integer in binary form. We want to ensure that if two inputs are encoded as the same Boolean vector evaluate to the same value. This requires rounding the thresholds in $T$.
\begin{lemma}
\label{lem:round}
    Let $T:[0,1]^n \rightarrow \{0,1\}$ be a balanced size-$s$ tree and $\varepsilon > 0 $. Then for $w= O(\log(s /\varepsilon))$
    let $T_\text{round}$ be the tree formed by rounding all thresholds of $T$ to the nearest $2^{-w}$.  Then
        $\dist(T,T_\text{round}) \leq \eps/2$.
\end{lemma}
\begin{proof}[Proof of~\Cref{lem:round}]
    Since $T$ and $T_\text{round}$ have the same label at every leaf, if $T(x) \neq T_\text{round}(x)$, it means that $x$ reaches different leaves in $T$ and $T_\text{round}$. Fix some leaf $\ell$ of  $T$. Then, $\ell$ is reached by some subcube of inputs of the form:
    \begin{align*}
        a_i \leq x_i < b_i \quad \text{for $i\in [n]$.}
    \end{align*}
    The corresponding leaf, $\ell_\text{round}$, of $T_\text{round}$ is reached by all inputs satisfying the following, where $\mathsf{round}(\cdot)$ rounds an input to the nearest multiple of $2^{-w}$:
    \begin{align*}
        \mathsf{round}(a_i) \leq x_i < \mathsf{round}(b_i) \quad \text{for $i \in [n]$.}
    \end{align*}
    We will upper bound the probability that a randomly chosen $x \in [0,1]^n$ reaches $\ell$ in $T$ but does not reach $\ell_\text{round}$ in $T_\text{round}$. For that to occur, there must be one $i \in [n]$ such that $x_i \in [a_i, \mathsf{round}(a_i)]$ or $x_i \in [\mathsf{round}(b_i), b_i]$. Since $|\mathsf{round}(z) - z| \leq 2^{-w}$, for any fixed $i$, $x_i$ falls in one of those ranges with probability at most $2 \cdot 2^{-w}$.
    
    Furthermore, since $T$ is balanced, $\ell$ has depth at most $O(\log(s))$. This means that for all but $O(\log(s))$ choices for $i$, $a_i = 0$ and $b_i = 1$, in which case there is no chance that $x_i \in [a_i, \mathsf{round}(a_i)]$ or $x_i \in [\mathsf{round}(b_i), b_i]$. By union bounding the up to $O(\log(s))$ input coordinates that matter for $\ell$, we have that the probability that $x$ reaches $\ell$ in $T$ but does not reach $\ell_\text{round}$ in $T_\text{round}$ is at most $O(\log(s) \cdot 2^{-w})$.
    
    By union bound over the $s$ leaves, the probability $x$ reaches a different leaf in $T$ as in $T_\text{round}$ is at most $O(s \log s \cdot 2^{-w})$. Setting this equal to $\eps/2$ and solving for $w$ yields the desired result.
\end{proof}

We are now ready to complete the proof of \Cref{thm:JZ real input}. Note that
\begin{align*}
    \bias(f) - \dist(f, T_\text{round}) 
    &\geq \bias(f) - (\dist(f,T) + \dist(T, T_\text{round})) \\
    & \geq \bias(f) - \dist(f,T) - \lfrac{\varepsilon}{2} \\
    & \geq \lfrac{\varepsilon}{2}.
\end{align*}

We will create a tree with Boolean inputs, $S: \{\pm 1\}^{wn} \rightarrow \{0,1\}$, satisfying the following relation
\begin{align*}
    S(E(x_1),...,E(x_n)) = T_\text{round}(x) \quad \text{for all $x \in [0,1]^n$}.
\end{align*}
The decision tree computing $S$ is the tree computing $T$ where each threshold is replaced with the tree of size at most $w$ specifying that threshold. Replacing each node of $T$ with a tree of size $w$ creates a tree that is a factor of $w^d$ larger than $T$, where $d$ is the depth of $T$. Using $w = O(\log (s/\varepsilon))$, we then have that $S$ has size 
\begin{align*}
    O(\log(s/\varepsilon))^d = 2^{O(\log(s) \log\log(s/\varepsilon))}.
\end{align*}
We are now ready to apply \Cref{thm:JZ}.  We have determined that $\bias(f) - \dist(f, T_\text{round}) \geq \eps/2$, and that $T_{\text{round}}$ is encodable as a decision tree of size $2^{\log(s) \log(\log(s/\varepsilon))}$ with Boolean inputs. By \Cref{thm:JZ}, there is a bit of the encoding with influence, with respect to $f$, at least $\Omega(\eps/\log(s)\log\log(s/\eps))$.

All that remains is to show that the most influential bit of the encoding, is a single threshold of the form $\Ind[x_i \geq \frac1{2}]$. All bits of the encoding for a variable, $x_i$, restrict $x_i$ to exactly half of the space in $[0,1]$. Since $f$ is monotone, the most influential such restriction is $\Ind[x_i \geq \frac1{2}]$. Hence, the split $\Ind[x_i \geq \frac1{2}]$ has influence at least $\Omega\left(\varepsilon/\log(s) \log\log(s/\eps)\right)$.  Since $f$ is monotone, this is also the correlation of $\Ind[x_i\ge \frac1{2}]$ with $f$, proving \Cref{thm:JZ real input}.
\end{proof}

\begin{proof}[Proof of~\Cref{thm:main-informal}]
The remainder of the proof of \Cref{thm:main-informal} is the same as the proof of \Cref{thm:main}.  Let $T^\circ$ be the size-$(j+1)$ partial tree built \BuildTopDownDT after $j$ iterations. As long as $\dist(f, T^\circ_f) > \balancedopt_s + \varepsilon$, we know there is a split that results in purity gain of at least\footnote{From \Cref{thm:JZ real input}, we actually know there is a split with this much purity gain that is the median of one input coordinate, although $\BuildTopDownDT$ may choose one that is not the median if it results in more purity gain.}
\begin{align*}
    \Omega\left(\frac{\kappa \cdot \varepsilon^2}{j \cdot (\log(s) \log\log(s/\varepsilon))^2}\right).
\end{align*}
Therefore, after we have run for
\[    t = 2^{O(\log(s) \log\log(s/\varepsilon))^2 /\kappa  \varepsilon^2}  = s^{\tilde{O}((\log s) /\kappa \varepsilon^2)}
\]
iterations we must have that $\dist(f, T^\circ_f) < \balancedopt_s + \varepsilon$, proving \Cref{thm:main-informal}.
\end{proof}
\section{Near-matching lower bound:~\Cref{thm:lower-bound}}

Our function $f$ witnessing the lower bound of~\Cref{thm:lower-bound} will be a variant of the well-known $\textsc{Tribes}_{\ell} : \zo^\ell \to \zo$ function from the analysis of boolean functions.  Recall that $\Tribes_{\ell}$ is a read-once DNF formula with $m\coloneqq \lfloor \frac{\ell}{w}\rfloor$ terms of width exactly $w$ over disjoint variables (with some variables possibly left unused), 
\[ \Tribes_\ell(x) \coloneqq T_1(x) \vee \cdots \vee T_m(x),  \] 
and $w \coloneqq \log\ell - \log \ln\ell + o_\ell(1)$ is chosen so that $\Pr[\Tribes_\ell(\bx) = 1]$ is as close to $\frac1{2}$ as possible.\footnote{While this acceptance probability cannot be made {\sl exactly} $\frac1{2}$ for all values of $\ell$ due to granularity issues, it will be the case that $\Pr[\Tribes_\ell(\bx)] = \frac1{2} \pm O(\frac{\log \ell}{\ell})$.  For clarity we will assume for the remainder of this proof the acceptance probability of $\Tribes_\ell$ is exactly $\frac1{2}$, noting that the same calculations go through if one carries around the additive $o_\ell(1)$ factor.}  (For more on the $\Tribes$ function, see Chapter \S4.2 of~\cite{ODbook}.)  Our construction will also involve the Majority function, $\Maj_k(y) \coloneqq \Ind[\sum_{i=1}^k y_i \ge 0]$.  

\pparagraph{The function that witnesses the separation.}  Let $m' < m$ be chosen so that the function $\Tribes_{\ell}' : \zo^{\ell} \to \zo$, 
\[ \Tribes_\ell'(x) \coloneqq T_1(x) \vee \cdots \vee T_{m'}(x), \] 
has acceptance probability $\Pr[\Tribes_\ell'(\bx) = 1]$ as close to $0.499$ as possible.\footnote{The same remark as in the previous footnote applies here.}  We additionally define $\Rest_{\ell} : \zo^\ell \to \zo$ to be:  
\[ \Rest_\ell(x) \coloneqq \Tribes_\ell(x) \wedge \neg\,\Tribes_\ell'(x),\] 
nothing that $\Pr[\Rest_\ell(\bx) = 1] = 0.001$. 

Our function $f_{\ell,k} : \zo^\ell \times \zo^k \to \zo$ is defined as follows: 
\[ f_{\ell,k}(x,y) \coloneqq
\begin{cases}
1 & \text{if $\Tribes_{\ell}'(x)=1$} \\
\Maj_k(y) & \text{if $\Rest_{\ell}(x) =1$}  \\
0 & \text{otherwise}, 
\end{cases} 
 \] 
 where $k$ and $\ell$ are chosen so that~\Cref{eq:k-and-ell} in the statement of~\Cref{lem:lower-bound-error} below is satisfied with equality.   
 
 \begin{proposition}[Upper bound on $\opt$]
 \label{prop:opt-upper-bound}
$\opt_{f_{\ell,k},2^{\ell}} \le 0.01$. 
 \end{proposition}

 \begin{proof} 
 We have that 
 \begin{align*}
  \Pr[f_{\ell,k}(\bx,\by)\ne \Tribes_\ell(\bx)] &= \Pr[\Rest_\ell(\bx) = 1 \text{ and } \Maj_k(\by) \ne \Tribes_\ell(\bx)] \\
  &= \Pr[\Rest_\ell(\bx) = 1] \cdot \Pr[\Maj_k(\by) = 0] \\
  &= 0.001 \cdot \lfrac1{2} \ < \ 0.01.
  \end{align*}  
 Since $\Tribes_\ell$ is computed by a decision tree of size $\le 2^\ell$, it follows that that $\opt_{f_{\ell,k},2^{\ell}} \le 0.01$.  
 \end{proof} 

\begin{lemma}[Lower bound on $\TopDownError$]
\label{lem:lower-bound-error} 
There are universal constants $c_1$ and $c_2$ such that the following holds.  Suppose $k$ and $\ell$ satisfy: 
\begin{equation}
\frac{c_1}{\sqrt{k}} \ge \frac{\log \ell}{\ell}. 
\label{eq:k-and-ell} 
\end{equation} 
Then $\TopDownError(f_{\ell,k},2^{\,c_2k/\log k}) \ge 0.49$. 
\end{lemma} 

\begin{proof}[Proof of~\Cref{thm:lower-bound} assuming~\Cref{lem:lower-bound-error}]
We choose $k$ and $\ell$ such that~\Cref{eq:k-and-ell} is satisfied with equality, and define $s \coloneqq 2^\ell$.  \Cref{thm:lower-bound} follows as an immediate consequence of~\Cref{prop:opt-upper-bound} and~\Cref{lem:lower-bound-error} since $2^{\Omega(k/\log k)} = 2^{\Omega(\ell^2/(\log \ell)^3)} = s^{\tilde{\Omega}(\log s)}$. 
\end{proof}

\subsection{Proof of~\Cref{lem:lower-bound-error}} 

Our proof of~\Cref{lem:lower-bound-error} draws on many of the ideas in~\cite{BLT20}'s proof of their Theorem 6(b).  Let $T$ denote the tree constructed by $\BuildTopDownDT(f_{\ell,k},2^{\,c_2k/\log k})$, where $c_2$ is a universal constant that will be determined later.  
 
 \begin{claim}
 \label{claim:small-influence} 
Let $\overline{\pi}$ be a path in $T$ that leads to a first query to an $x$-variable.  Then 
$\Inf_{y_j}(\Maj_k(y)_{\overline{\pi}}) \le \frac1{100\sqrt{k}}$ for all $j\in [k]$.  
 \end{claim}
 
 \begin{proof}
 Suppose without loss of generality that $\overline{\pi}$ leads to a query to $x_1$.   By the splitting criterion of $\BuildTopDownDT$ and~\Cref{prop:influence-correlation-monotone}, we have that 
 \[ \Inf_{x_1}(f_{\overline{\pi}}) \ge \Inf_{y_j}(f_{\overline{\pi}}) \quad \text{for all $j\in [k]$}. \] 
 Since
 \[ \Inf_{x_1}(f_{\overline{\pi}}) \le \Inf_{x_1}(\Tribes_{\ell}) + \Inf_{x_1}(\Tribes_\ell') \le O\left(\frac{\log \ell}{\ell}\right), \]
and 
 \[ \Inf_{y_j}(f_{\overline{\pi}}) = \Pr[\Rest_\ell(\bx)=1] \cdot \Inf_{y_j}(\Maj_k(y)_{\overline{\pi}}) = 0.001 \cdot \Inf_{y_j}(\Maj_k(y)_{\overline{\pi}}), \]
 it follows that $\Inf_{y_j}(\Maj_k(y)_\pi) \le O(\frac{\log \ell}{\ell})$. The claim follows by choosing $c_1$ to be a sufficiently small constant in~\Cref{eq:k-and-ell}.
 \end{proof} 
 
 \begin{corollary}
 \label{cor:deep}
There is a universal constant $c_3$ such that the following holds.  Let $(\bx,\by)\sim \zo^k \times \zo^\ell$ be a uniform random input, and $\pi_{(\bx,\by)}$ be the corresponding root-to-leaf path in $T$ that $(\bx,\by)$ follows.  The probability that $\pi_{(\bx,\by)}$ queries an $x$-variable before at least $c_3k/\log k$ many $y$-variables is at most $0.001$. 
 \end{corollary} 
 
 \begin{proof} 
 Call an input $(x,y)$ {\sl bad} if $\pi_{(x,y)}$ queries an $x$-variable before $c_3k$ many $y$-variables.  Let $\overline{\pi}$ denote the truncation of $\pi_{(x,y)}$ to its prefix before the first query to an $x$-variable.  By~\Cref{claim:small-influence}, we have that $\Inf_{y_j}(\Maj_k(y)_{\overline{\pi}})\le \frac1{100\sqrt{k}}$, and so the discrepancy between the number of $0$'s and $1$'s in $\overline{\pi}$ must be $\Omega(\sqrt{k})$.   Therefore, we can bound 
 \begin{align*} \Pr[(\bx,\by)\text{ is bad}] &\le \sum_{t=1}^{c_3k/\log k} \Prx_{\bb\sim \mathrm{Bin}(t,\frac1{2})}\big[|\bb-\lfrac{t}{2}| \ge \Omega(\sqrt{k})\big]  \\
 &\le \sum_{t=1}^{c_3k/\log k} e^{-\Theta(k/t)} \tag*{(Hoeffding's inequality)} \\
 &\le c_3k \cdot e^{-\Theta(\log k/c_3)} \ \le \ o_k(1). 
 \end{align*}
where the final inequality holds by choosing $c_3$ to be a sufficiently small constant. 
 \end{proof}

We are now ready to prove~\Cref{lem:lower-bound-error}.  Let $\xi : \zo^k \times \zo^\ell \to \zo$ be the indicator
\[ \xi(x,y) \coloneqq \Ind\Big[\text{$\pi_{(x,y)}$ queries an $x$-variable before at least $\frac{c_3k}{\log k}$ many $y$-variables}\Big].\] 
By~\Cref{cor:deep} we have that $\Pr[\xi(\bx,\by)=1] \le 0.001$.  If $(x,y)$ is such that $\Pr[\xi(x,y) =0]$, then either
\begin{enumerate}
\item $\ds |\pi_{(x,y)}|\ge \frac{c_3k}{\log k}$, or 
\item $\ds |\pi_{(x,y)}|< \frac{c_3k}{\log k}$ and $\pi_{(x,y)}$ does not query any $x$-variables. 
\end{enumerate} 
Since the fraction of inputs that follow any specific path of length $\ge c_3k/\log k$ is at most $2^{-c_3k/\log k}$, and the size of $T$ is $2^{\,c_2k/\log k}$ by assumption, choosing $c_2 = \frac1{2} c_3$ ensures that the fraction of inputs $(x,y)$ such that $\pi_{(x,y)}$ falls into the first case above (i.e.~$|\pi_{(x,y)}|\ge c_3k/\log k$) is at most $0.001$.

Therefore, at least a $0.998$ fraction of inputs $(x,y)$ are such that $\pi_{(x,y)}$ falls into the second case above.  For such $(x,y)$'s, we have that 
\begin{align*}
 \Pr[f_{\pi_{(x,y)}}=1] &\ge \Pr[\Tribes_\ell' = 1] =  0.499 \\
 \Pr[f_{\pi_{(x,y)}}=0] &\ge \Pr[\Tribes_\ell = 0] = 0.5,
 \end{align*}
 and so $\Pr[T_{\pi_{(x,y)}} \ne f_{\pi_{(x,y)}}] \ge 0.49$.  We conclude that $\Pr[T(\bx,\by)\ne f(\bx,\by)] \ge 0.998 \times 0.499 > 0.49$, which completes the proof of~\Cref{lem:lower-bound-error}.

\section{Revisiting and strengthening~\cite{BLT20}'s guarantees}
\label{section: realizable setting}
\cite{BLT20}'s work on the realizable setting analyzed the performance of a {\sl variant} of the top-down heuristics; their variant does not correspond to any impurity function $\mathscr{G}$.  

Consider $\BuildTopDownDT_\Inf$, defined in \Cref{fig:TopDownBLT20}.

\begin{figure}[H]
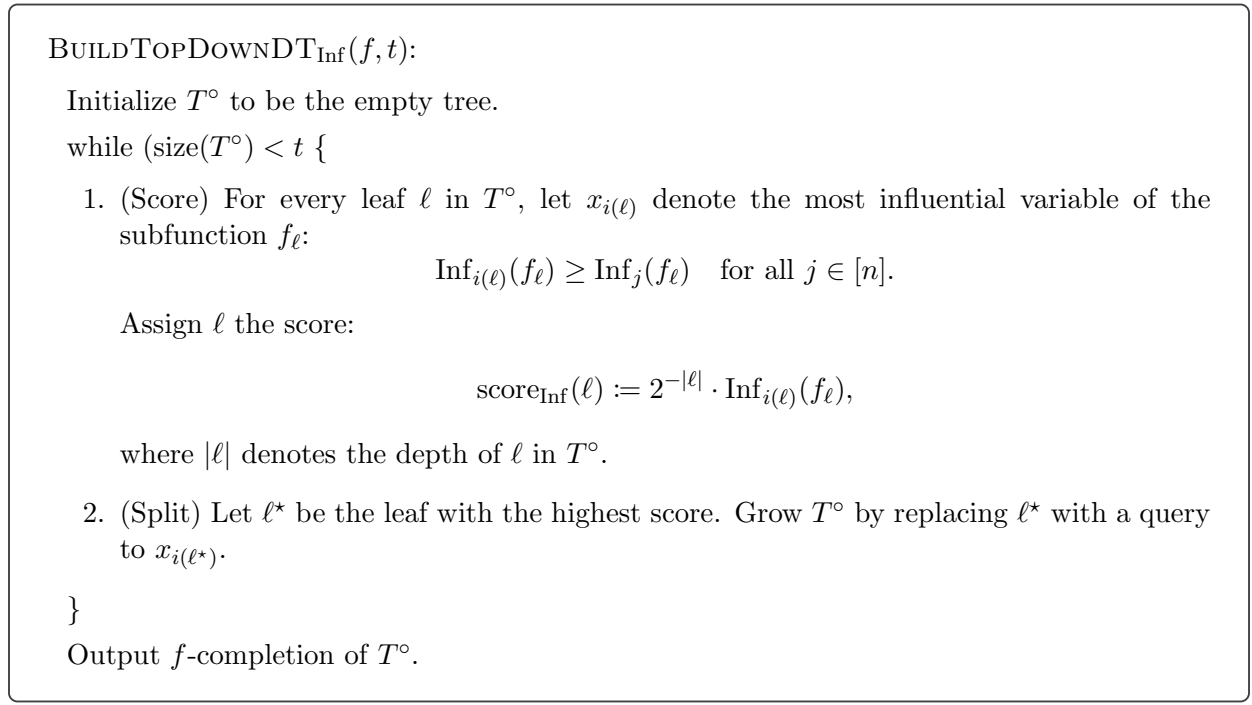

  \captionsetup{width=.9\linewidth}
\begin{tcolorbox}[colback = white,arc=1mm, boxrule=0.25mm]
\vspace{3pt} 

$\BuildTopDownDT_\Inf(f,t)$:  \vspace{6pt} 

\ \ Initialize $T^\circ$ to be the empty tree. \vspace{4pt} 

\ \ while ($\size(T^\circ) < t$ \{
\vspace{-3pt} 
\begin{enumerate}
\item {(Score)} For every leaf $\ell$ in $T^\circ$, let $x_{i(\ell)}$ denote the most influential variable of the subfunction $f_\ell$: 
\[ \Inf_{i(\ell)}(f_\ell) \ge \Inf_j(f_\ell) \quad \text{for all $j\in [n]$.} \] 
Assign $\ell$ the score: 
\begin{align*}
 \mathrm{score}_\Inf(\ell) &\coloneqq 2^{-|\ell|} \cdot \Inf_{i(\ell)}(f_\ell),
 \end{align*} 
 where $|\ell|$ denotes the depth of $\ell$ in $T^\circ$. 
 \item {(Split)} Let $\ell^\star$ be the leaf with the highest score.  Grow $T^{\circ}$ by replacing $\ell^\star$ with a query to $x_{i(\ell^\star)}$. 
\end{enumerate}
\ \ \}  \vspace{4pt} 

\ \ Output $f$-completion of $T^\circ$.
\vspace{3pt}

\end{tcolorbox}
\caption{\cite{BLT20}'s top-down heuristic for building a decision tree approximation for~$f$.}
\label{fig:TopDownBLT20}
\end{figure}

\begin{lemma}[Theorem 5 of \cite{BLT20}]
    \label{lemma: BLT 20 upper bound}
    For every $\varepsilon \in (0, \frac{1}{2})$ and monotone function $f: \{\pm 1\}^n \to \{0,1\}$ exactly computable by a size $s$ decision tree,
    \begin{align*}
        \TopDownError_\Inf(f, s^{O(\sqrt{\log s} / \varepsilon)}) \leq \varepsilon.
    \end{align*}
\end{lemma}

We will show that as a consequence of \Cref{lemma: BLT 20 upper bound}, a similar upper bound can be shown for $\BuildTopDownDT_\mathscr{G}$ for any strongly-concave $\mathscr{G}$. We first briefly summarize the proof of \Cref{lemma: BLT 20 upper bound} given in \cite{BLT20}. They define the following potential function:
\begin{align*}
    u_f(T^\circ) \coloneqq \sum_{\text{leaves $\ell \in T^\circ$}} 2^{-|\ell|} \cdot \Inf(f_\ell).
\end{align*} 
For monotone functions $f: \{\pm 1\}^n \to \{0,1\}$ that are computable by size-$s$ decision trees, a result of~\cite{OS07} gives the following bound: 
\begin{align*}
    u_f(\text{empty tree}) = \Inf(f) \leq \sqrt{\log s}.
\end{align*}
Furthermore, they show that if $T^\circ$ is the size-$(j+1)$ tree built after $j$ iterations of $\BuildTopDownDT_\Inf(f)$, and $\dist(f, T^\circ_f) \geq \varepsilon$, then the leaf, $\ell^\star$, selected in the next iteration satisfies.
\begin{align*}
    \score(\ell^\star) \geq \frac{\varepsilon}{(j+1) \log s}.
\end{align*}
Since $u_f(T^\circ)$ decreases by the score of the leaf selected, and $u_f(T^\circ)$ upper bounds $\dist(f, T^\circ_f)$,~\cite{BLT20} are able to conclude that a growing a decision tree of size
\begin{align*}
    2^{O(\sqrt{\log s } \cdot \log(s) / \varepsilon)} = s^{O(\sqrt{\log s}/\varepsilon)}
\end{align*}
suffices to ensure an error of $\le \eps$. Using this same proof outline, we will establish the same guarantee for $\BuildTopDownDT_\mathscr{G}$ for any strongly concave impurity function $\mathscr{G}$: 
\begin{theorem}[Extending Theorem 5 of~\cite{BLT20} to actual top-down heuristics]
    \label{thm: BLT1 for impurity}
    Let $\mathscr{G}$ be any $\kappa$-strongly concave impurity function, $\varepsilon \in (0, \frac{1}{2})$, and $f: \{\pm 1\}^n \rightarrow \{0,1\}$ be a monotone function computable by a size-$s$ decision tree. Then,
    \begin{align*}
        \TopDownError_\mathscr{G}(f, s^{O(\sqrt{\log s} / \eps)}/\kappa) \leq \eps.
    \end{align*}
\end{theorem}
\begin{proof}Recalling~\Cref{prop: score reduction}, for any leaf $\ell \in T^\circ$, the variable with maximum influence in~$f_\ell$ will also be the variable that results in the maximum purity gain when split. Therefore, at any leaf, $\BuildTopDownDT_\mathscr{G}$ will always split the same variable as $\BuildTopDownDT_\Inf$, but may just choose a different order of leaves to split. Splitting extra leaves can only decrease the error, so once $\BuildTopDownDT_\mathscr{G}$ has split every leaf that $\BuildTopDownDT_\Inf$ would in $s^{O(\sqrt{\log s}/\varepsilon)}$ iterations, the resulting tree must have error less than $\varepsilon$.
    
    We know that running $\BuildTopDownDT_\Inf$ for $s^{O(\sqrt{\log s}/\varepsilon)}$ iterations is sufficient to ensure an error of at most $\eps$. Let $T^\circ$ be the size-$(j+1)$ tree built after $j$ iterations of $\BuildTopDownDT_\Inf(f)$. \cite{BLT20} proved that if $\dist(f, T^\circ_f) \geq \varepsilon$, then the leaf, $\ell^\star$, selected in the next iteration satisfies: 
    \begin{align*}
        \score(\ell^\star) \geq \frac{\varepsilon}{(j+1) \log s}.
    \end{align*}
    Therefore, we can substitute in $j = s^{O(\sqrt{\log s}/\varepsilon)}$ to find that if $\BuildTopDownDT_\Inf$ has split all leaves with score at least
    \begin{align*}
        \frac{\varepsilon}{s^{O(\sqrt{\log s}/\varepsilon)} \cdot\log s} = \frac{1}{s^{O(\sqrt{\log s}/\varepsilon)}},
    \end{align*}
    the tree it has built must have error $\varepsilon$. We will show that $\BuildTopDownDT_\mathscr{G}$ will not take too long to split any leaf with score at least $1 / s^{O(\sqrt{\log s}/\varepsilon)}$, therefore proving an upper bound on the number of iterations it needs to reach error $\varepsilon$. Let $\ell$ be any leaf with score at least $1 / s^{O(\sqrt{\log s}/\varepsilon)}$ and $i$ be the index of its most influential variable. Then, 
    \begin{align*}
        \puritygain_f(T^\circ,\ell,x_i) &\geq 2^{|-\ell|} \cdot \frac{\kappa}{32} \cdot \Inf_i(f)^2 \tag*{(\Cref{lem:useful-properties})}\\
        & \geq \frac{\kappa}{32} \cdot (2^{|-\ell|} \cdot \Inf_i(f))^2 \\
        & = \frac{\kappa}{32} \cdot  \score(\ell)^2 \\
        &\geq  \frac{\kappa}{s^{O(\sqrt{\log s}/\varepsilon)}}.
    \end{align*}
     Let $\ell^\star$ be some leaf with larger purity gain that $\ell$, and $i^\star$ be the associated variable that is split. Then, since $\puritygain_f(T^\circ,\ell^*,x_{i^\star}) \leq 2^{-|\ell^\star|}$, we must have that $|\ell^\star| < \log(s^{O(\sqrt{\log s}/\varepsilon)}/\kappa)$. There are at most $s^{O(\sqrt{\log s}/\varepsilon)}/\kappa$ possible such nodes, so after constructing a tree of size $s^{O(\sqrt{\log s}/\varepsilon)}/\kappa$, $\BuildTopDownDT_\mathscr{G}$ must split $\ell$, and we conclude that it must have achieved error at most~$\varepsilon$.
\end{proof}

\section{Conclusion} 

We have given strengthened provable guarantees on the performance of widely employed and empirically successful top-down decision tree learning heuristics such as ID3, C4.5, and CART.  Compared to previous works, our guarantees: (1) hold in the more realistic and challenging agnostic setting; (2) apply to all top-down heuristics and their associated impurity functions; (3) extend to the setting of real-valued features and arbitrary product distributions over the domain.  Our main result shows that for all monotone functions $f : \R^n \to \zo$ and $s \in \N$, these top-down heuristics build a tree of size $s^{\tilde{O}((\log s)/\eps^2)}$ that achieves error within $\eps$ of that of the optimal balanced size-$s$ decision tree for $f$.   We complement this with a near-matching lower bound.  While our work was primarily motivated by the goal of understanding top-down heuristics, our results yield new guarantees that are not known to be achievable by any other algorithm, even ones that are not based on top-down heuristics.   

There are several concrete avenues for future work:  
\begin{enumerate}
\item {\sl Beyond monotonicity.} As mentioned in the introduction, any top-down heuristic will fare badly on the parity functions $f$, in the sense of building a tree that is much larger than the optimal tree for $f$. Though broad and natural, the class of monotone functions is not the only class that excludes the parity function. Another fundamental property to consider is {\sl noise stability} (see \S2.4 of~\cite{ODbook})---what guarantees can be made about the performance of these top-down heuristics when run on noise-stable functions? 
\item {\sl Beyond product distributions.} In this work our results hold for arbitrary product distributions over the domain, extending previous work that focuses on the uniform distribution. Could we establish provable distribution-independent guarantees, or failing that, perhaps provable guarantees for distributions with limited dependencies between coordinates?
\item {\sl Polynomial-size approximating trees.} Our lower bound (\Cref{thm:lower-bound}) shows a monotone function such that any top-down heuristics has to build a tree of size $s^{\tilde{\Omega}(\log s)}$ in order to achieve error $\le \opt_s + \eps$.~\cite{BLT20} show a similar lower bound of $s^{\tilde{\Omega}(\sqrt[4]{\log s})}$ in the realizable setting. Are there broad and natural {\sl subclasses} of monotone functions that evade these lower bounds, and for which polynomial size upper bounds do exist?

\end{enumerate} 

\section*{Acknowledgements} 

We thank Michael Kim and the ICML reviewers for their helpful feedback and suggestions.   LYT is supported by NSF grant CCF-192179 and NSF CAREER award CCF-1942123.


\bibliography{most-influential}
\bibliographystyle{alpha}

\end{document}